\documentclass[aps,prb,reprint,groupedaddress, longbibliography]{revtex4-1}

\usepackage{amssymb,amsmath,latexsym,amscd,amsfonts, bbm,bm}
\usepackage{graphicx}  
\usepackage[usenames, pdftex, dvipsnames]{color}
\usepackage{umoline}
  
\newtheorem{theorem}{Theorem}[section]

\newtheorem{deff}[theorem]{Definition}  
  
\newtheorem{claim}[theorem]{Claim}  
\newtheorem{lem}[theorem]{Lemma}

\newtheorem{corol}[theorem]{Corollary}  
\newtheorem{fact}[theorem]{Fact}

 \newcommand{\qedsymb}{\hfill{\rule{2mm}{2mm}}}  
 \newenvironment{proof}[1][]{\begin{trivlist}  
 \item[\hspace{\labelsep}{\bf\noindent Proof#1:\/}] 
 }{\qedsymb\end{trivlist}}

\newcommand{\ignore}[1]{}

\newcommand{\norm}[1]{{\| #1 \|}}  
  
\newcommand{\ket}[1]{{ |{#1} \rangle }}  

\newcommand{\braket}[2]{{ \langle {#1} | {#2} \rangle}}

\newcommand{\orderof}[1]{\mathcal{O}(#1)} 
\newcommand{\poly}{\mathrm{poly}} 
\newcommand{\EqDef}{\stackrel{\mathrm{def}}{=}}

\newcommand{\Eq}[1]{Eq.~(\ref{#1})}
\newcommand{\Fig}[1]{Fig.~\ref{#1}}
\newcommand{\Def}[1]{Def.~\ref{#1}}
\newcommand{\Lem}[1]{Lemma~\ref{#1}}

\newcommand{\Sec}[1]{Sec.~\ref{#1}}
\newcommand{\Ref}[1]{Ref.~\onlinecite{#1}}
\newcommand{\Refs}[1]{Refs.~\onlinecite{#1}}

\newcommand{\Thm}[1]{Theorem~\ref{#1}}

\newcommand{\Id}{\mathbbm{1}}
\newcommand{\BB}{\mathbbm{B}}
\newcommand{\BN}{\mathbbm{N}}

\newcommand{\etal}{\textit{et al.}}

\newcommand{\gs}{\Omega}
\newcommand{\gsp}{\Omega^\perp}
\newcommand{\hPi}{\hat{\Pi}}

\newcommand{\hA}{\hat{A}}
\newcommand{\hD}{\hat{D}}
\newcommand{\hDel}{\hat{\Delta}}

\newcommand{\SR}[1]{\mathrm{SR}(#1)}
\newcommand{\Hi}{\mathcal{H}}
\newcommand{\Hip}{\mathcal{H^\perp}}
\newcommand{\Bs}{\bm{s}}
\newcommand{\Pyr}{\mathrm{Pyr}}

\begin{document}


\title{ An improved 1D area law for frustration-free
                     systems}


\author{Itai Arad}
\email{arad.itai@gmail.com}
\affiliation{School of Computer Science and Engineering,\\
 The Hebrew University,
  Jerusalem, 91904, Israel}
\author{Zeph Landau}
  \email{zeph.landau@gmail.com}
\author{Umesh Vazirani}
  \email{vazirani@eecs.berkeley.edu}
\affiliation{University of California at Berkeley\\
    Berkeley, Californina, 94720, USA}


\date{\today}

\begin{abstract}
  We present a new proof for the 1D area law for frustration-free
  systems with a constant gap, which exponentially improves the 
  entropy bound in Hastings’ 1D area law and which is tight to 
  within a polynomial factor.  For particles of dimension $d$, 
  spectral gap $\epsilon>0$, and interaction strength of, at most, $J$,
  our entropy bound is $S_{1D}\le \orderof{1}\cdot X^3\log^8 X$,
  where $X\EqDef(J\log d)/\epsilon$. Our proof is completely
  combinatorial, combining the detectability lemma with basic tools
  from approximation theory. In higher dimensions, when the
  bipartitioning area is $|\partial L|$, we use an additional local
  structure in the proof and show that $S \le
  \orderof{1}\cdot|\partial L|^2\log^6|\partial L|\cdot X^3\log^8X$.
  This is at the cusp of being nontrivial in the 2D case, in the
  sense that any further improvement would yield a subvolume law.

\end{abstract}


\maketitle

\section{Introduction}
\label{sec:Intro}

One of the striking differences between quantum and classical
systems is the number of parameters that are needed to describe
them. A classical system of $n$ particles can be generally described
by $\orderof{n}$ parameters, whereas an \emph{arbitrary} state of a
similar quantum system would generally require $2^{\orderof{n}}$
parameters. This exponential gap is directly related to the
phenomena of entanglement; quantum states do not have to be simple
product states but can be an arbitrary superposition of such
states.

But how genuine is this exponential gap? Is it an artifact of the
fact that we are considering \emph{arbitrary} quantum states, or is
it an inherent characteristic of \emph{physical} states that occur
in nature, which are, of course, a much more restricted set of
states? Among the best physical systems one may consider with
respect to this question are quantum many body systems on a lattice,
which are ubiquitous in condensed matter physics. These systems are
generally described by a \emph{local Hamiltonian} that models the
local interaction between neighboring particles. In particular, one
is interested in the entanglement properties of the \emph{ground
state} of these systems. If the system has a spectral gap
$\epsilon>0$, we can approach this state by cooling the system to
temperatures below that gap. Such cold systems would be in an almost
perfectly coherent state, and one expects their quantum nature to be
fleshed out most pronouncedly.  What are then the entanglement
properties of such states?

Area laws constitute one of the most important tools for bounding
entanglement in such systems.\cite{ref:AL-rev} Starting from
Bekenstein's seminal result that the entropy of a black hole is
proportional to the surface of its horizon, \cite{ref:black-hole73}
it was later conjectured that the origin of this entropy is the
quantum entanglement between the inner part of the black hole and
its surrounding.\cite{ref:Bombelli86, ref:Srednicki93} This
conjecture has led researchers to consider the scaling of
entanglement entropy in the ground (vacuum) states in models of
quantum field theories,\cite{ref:Kabat94, ref:Holzhey94} where it
was demonstrated that also in there, entropy scales like the
\emph{surface area} of the region rather than its volume, albeit
with some logarithmic corrections in critical cases. The same
behavior was then demonstrated also in the context of spin chains by
Vidal \etal,\cite{ref:AL-Vidal03} which has led to the formulation
of the area-law conjecture that we now describe.

Consider a local-Hamiltonian system of $d$-dimensional
particles that sit on a $D$-dimensional grid, as demonstrated in
\Fig{fig:commuting-area-law}.  We say that a state $\ket{\psi}$ of
this system obeys an area law if, for any contiguous region $L$ of
the grid, the entanglement between the particles inside $L$ and
the particles outside $L$ is upper bounded by a constant times the
surface area of $L$ as follows: 
\begin{align}
 S_L(\psi) \le \orderof{|\partial L|} \ .
\end{align}
The area-law conjecture then states that if the system has a
constant spectral gap $\epsilon>0$, its \emph{ground state}
necessarily obeys an area law. This is clearly much stronger than
the trivial bound on the entropy (known as a \emph{volume law}),
which is proportional to the number of particles 
inside $L$. 
\begin{figure}
 \begin{center}
   \includegraphics[scale=1]{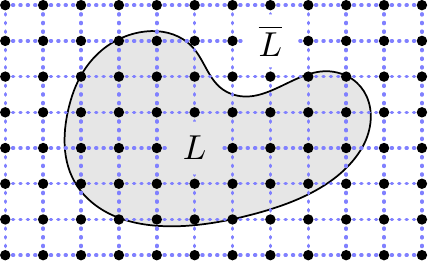}
 \end{center}
 \caption{
 An illustration of a quantum many-body system on a grid. The
 particles sit on the vertices and the edges denote 2-body
 interactions. When the system is described by a state that obeys
 an area law, the entanglement entropy between the particles inside
 the region $L$ and the particles outside of it is proportional to
 the surface area $|\partial L|$.
 \label{fig:commuting-area-law} }
\end{figure}

On the face of it, there are several reasons to believe that ground
states of gapped local Hamiltonians obey an area law. One might
intuitively expect that the entanglement is local because it is
generated by local interactions. More formally, it has been shown by
Hastings that in the ground state of gapped systems, the correlation
between two local observables decays exponentially in their lattice
distance.\cite{ref:Has04} It is, therefore, tempting to conclude
that only the degrees of freedom near the boundary are entangled to
those outside the region. However, the existence of data hiding
states \cite{ref:data-hiding} implies that such simple reasoning is
probably insufficient to prove an area law.\cite{ref:Has07b}

Indeed, even though the area-law conjecture was shown to be true in
many specific models, mostly in 1D (see, for example,
\Refs{ref:AL-Vidal03, ref:Kor1, ref:Kor2, ref:AL-rev} and references
therein), it was not until a few years ago, in a seminal
paper,\cite{ref:Has07} that Hastings proved that it holds for
\emph{all} 1D systems with a spectral gap. In this case, an area law
says that ground-state entropy across any cut in the 1D chain is
bounded by a constant, independent of the system size. From this,
one can deduce that, for any practical purposes, the ground state of
such systems can be described by a \emph{polynomial} number of
parameters (say, by a Matrix Product State), instead of an
exponential number (see \Ref{ref:Has07}).

Unfortunately, Hastings' 1D proof does not scale up to higher
dimensions. One reason lies in the form of the entropy bound. For a
1D system with a bounded interaction strength $J$, a spectral gap
$\epsilon>0$, and a particle dimension $d$, the upperbound is
$2^{\orderof{X}}$, where $X\EqDef \frac{J\log d}{\epsilon}$. This
exponential dependence on $\log d$ is catastrophic if we want to use
this formula in higher dimensions by fusing together particles along
surfaces parallel to $\partial L$. It implies an entropy bound of
$2^{\orderof{|\partial L|}}$ instead of $\orderof{|\partial L|}$.
This is exponentially larger than the trivial volume-law bound of
$\orderof{|L|}$. Proving area laws in two or higher dimensions
therefore is still a wide open problem and a holy grail in quantum
Hamiltonian complexity.

A combinatorial approach to proving the area law for 1D frustration
free systems was introduced in \Ref{ref:DL2}. The proof replaced
Hastings' analytical machinery including the Lieb-Robinson bound and
spectral Fourier analysis with the detectability lemma,\cite{ref:DL}
a combinatorial lemma about local Hamiltonians. However, the
resulting bound was no better than Hastings' bound because, at its
heart, the argument followed the same outline as Hastings',
including the use of a ``mutual information saturation''-type
argument (or, in effect, a kind of ``monogamy of entanglement''-type
argument) that leads to an exponential slack. Nevertheless, the
combinatorial nature of the detectability lemma opened up the
possibility of a new inherently combinatorial proof of the area law.

\subsection{Our results}

In this paper we give a new proof of the area law in 1D for the case
of gapped frustration-free Hamiltonians. Our proof yields an
exponentially better bound on the entanglement entropy than the
bound in \Ref{ref:Has07}. Specifically, we prove
\begin{theorem}
 \label{thm:main}
 Consider a 1D, frustration-free local Hamiltonian system
 $H=\sum_i H_i$, over a system of $n$ particles of local
 dimension $d$, with $H_i$ being 2-local nearest-neighbor
 interactions with a bounded strength $\norm{H_i}\le J$.
 Assume further that the system has a unique ground
 state $\ket{\gs}$ and a spectral gap $\epsilon>0$, and define
 \begin{align}
 \label{def:X}
   X \EqDef \frac{J\log d}{\epsilon} \ .
 \end{align}
 Then along any cut in the chain, the von-Neumann entanglement
 entropy between the two sides is bounded by
 \begin{align}
   \label{eq:S}
   S_{1D} \le \orderof{1}\cdot X^3\log^8X \ .
 \end{align}
\end{theorem}

The importance of this work is threefold. First, in 1D, we
exponentially improve the upper bound as a function of
$\epsilon^{-1}$: the bound in \Eq{eq:S} implies an upper bound of
$\orderof{\epsilon^{-3}\log^8{\epsilon^{-1}}}$. This is within a
polynomial factor of recent lower bounds that were found by Hastings
and Gottesman \cite{ref:lowS-2} and Irani.\cite{ref:lowS-1}\ \
\footnote{The examples in these papers are of a frustrated
Hamiltonians, whereas our results apply for frustration-free
systems. However, the construction in \Ref{ref:lowS-1} is
translation-invariant. It can be easily modified to be made
frustration-free \cite{ref:Sandy-comment} if the
translation-invariance requirement is dropped, yielding a system in
which the entanglement entropy $S$ is lower-bounded by
$\epsilon^{1/6}$, where $\epsilon>0$ is the spectral gap of the
system.} The former showed a 1D system with a fixed $d$ and $J$, in
which the entropy along a cut is at least
$\orderof{\epsilon^{-1/4}}$, and the latter showed a lower bound of
$\orderof{\epsilon^{-1/12}}$ for a translation invariant system.

Second, and more important, is the relevance of this work for
proving area laws in higher dimensions. A naive application of our
bound to higher dimensional systems would yield an entropy bound of
$S\le |\partial L|^3\poly(\log|\partial L|)$, which is still worse
than a volume law, but is much better than the previous exponential
bound. Moreover, as we show in \Sec{sec:2D}, in two or more
dimensions, one can further exploit the local properties of the
system \emph{along the boundary $\partial L$} and improve the bound
to $|\partial L|^2\cdot\poly(\log|\partial L|)$. This bound is at
the cusp of being nontrivial; any further improvement that would
bound the entropy by $|\partial L|^{2-\delta}$ for any $\delta>0$,
would prove a subvolume law for 2D. In fact, this gives a
subvolume law for the case of fractals with dimension $1<D<2$, but
we shall not pursue that direction here.

Finally, our approach is differs from Hastings' original proof.
Here, the ``monogamy of entanglement''-type argument is replaced with an
iterative procedure to find product states with increasingly higher
overlap with the ground state. The central quantity used to analyze
the progress of this procedure is Schmidt rank, rather than more
advanced tools like relative entropy and its monotonicity. While the
iterative procedure is based on the detectability lemma, a much
more intricate combinatorial structure is necessary to ensure that
the procedure makes progress.

We now give a high-level overview of the proof.

\subsection{High-level overview of the proof}

Consider a 1D chain of $n$ $d$-dimensional particles with nearest
neighbor interactions, described by the gapped, frustration-free
Hamiltonian $H = \sum_i Q_i$ with a unique ground state $\ket{\gs}$.
For the sake of simplicity, we assume that $Q_i$ are projections
and, therefore, $P_i\EqDef \Id-Q_i$ are projections to the local
ground spaces of the different terms.

The key to proving an area law across a cut, is to find a product
state $\ket{\phi}=\ket{\phi_L}\otimes\ket{\phi_{\overline{L}}}$ with
respect to the bipartitioning of the system, which has a large overlap
with $\ket{\gs}$. Our approach to finding such a product state is to
start with any product state with a nonzero overlap with $\ket{\gs}$
and act on it with an operator that increases its overlap with
$\ket{\gs}$, without increasing its Schmidt rank much. Specifically,
we construct an operator $K$ with the following property: $K$ fixes
$\ket{\gs}$, but when applied to any state $\ket{\psi}$, it shrinks
the component orthogonal to $\ket{\gs}$ by a factor of $\Delta$ while
increasing the Schmidt rank of $\ket{\psi}$ by, at most, a factor of
$D$. Clearly, there is a race between these two factors $D$ and
$\Delta$. It turns out that when $D\cdot\Delta <1/2$, we can amplify
the overlap with $\ket{\gs}$ by replacing
$\ket{\phi}=\ket{\phi_L}\otimes\ket{\phi_{\overline{L}}}$ with one of
the Schmidt vectors of $K\ket{\phi}$. This amplification continues all
the way until the overlap is $\sqrt{1/(2D)}$. A few more applications
of $K$ to this product state yield a state with Schmidt rank
$D^{\orderof{1}}$, which has constant ($D$ independent) overlap with
$\ket{\gs}$. Further applications of $K$ give rise to Schmidt
coefficients with vanishing mass, and, therefore, the entanglement
entropy of $\ket{\gs}$ can be bounded by $\orderof{\log D}$.

The task of proving an area law therefore is reduced to the task of
finding an operator $K$ with $D\cdot\Delta<1/2$. Our starting point
is the detectability lemma (DL). Denote the spectral gap by
$\epsilon>0$. We can partition the projections $\{P_i\}$ into two
subsets of even and odd projections, which are called ``layers''
(see \Fig{fig:2layers}). Inside each layer, the projections commute
because they are nonintersecting. Consequently, $\Pi_{odd}\EqDef
P_1\cdot P_3\cdot P_5\cdots$ and $\Pi_{even}\EqDef P_2\cdot P_4\cdot
P_6\cdots$ are the projections into the common ground spaces of the
odd and even layers, respectively. Then, by the DL, the operator
$A\EqDef \Pi_{even}\Pi_{odd}$ is an \emph{approximation} to the
ground-state projection. It preserves the ground state while
shrinking its perpendicular subspace by an $n$-independent factor
$\Delta_0(\epsilon)\simeq 1-c\epsilon$, where $c$ is a geometrical
factor. Moreover, each application of $A$ increases the Schmidt rank
of our state by a constant factor of, at most, $D_0\EqDef d^2$ (due to
the projection that intersects with the cut in the chain; see
\Fig{fig:2layers}). Unfortunately, we generally expect $D_0
\cdot\Delta_0 > 1$, so the operator $A$ does not, by itself, suffice
to carry out our plan.

To construct the operator $K$ we need several new ideas. First, we
observe that $D_0$ and $\Delta_0$ can be replaced by $D_0^k$ and
$\Delta_0^k$, respectively, by coarse graining: Fuse $k$ adjacent
particles, making them a single particle of dimension $d^k$.
Although this only increases the value of the product, it creates
room for the next step, which is to modify the operator $A$ to
decrease the factor by which it blows up the Schmidt rank. For
concreteness, assume that the even layer contains the projection
that intersects with the cut. We will focus on a segment of $m$
projections around the cut and denote their product by $\Pi_m$, so
$\Pi_{even} = \Pi_m \Pi_{rest}$. We will replace the operator $\Pi_{even}$
with $\hPi_m\Pi_{rest}$ that closely approximates $\Pi_{even}$
while increasing the Schmidt rank by much less than $D_0^k$
(when amortized over several applications).

One of the great benefits of using the DL is that the all
projections in a given layer commute, and, hence, much of the
following analysis becomes almost classical.  Indeed, the $m$
projections around the cut $\{P_i\}_{i=1}^m$ define a decomposition of the
Hilbert space of the system into a direct sum of $2^m$ eigenspaces,
called sectors. Each sector is defined by a string $\Bs=(s_1,
\ldots, s_m)$, such that if $\ket{\psi_{\Bs}}$ is in the $\Bs$ sector,
$P_i\ket{\psi_{\Bs}} = (1-s_i)\ket{\psi_{\Bs}}$. A site with $s_i=1$ is
called \emph{a violation}, since it corresponds to a non-zero energy
of the corresponding local Hamiltonian term, and $\sum_{i=1}^m s_i$
is the total number of violations in the sector $\Bs$.

Now, an arbitrary state $\ket{\psi}$ can be decomposed as
$\ket{\psi}=\ket{\psi_0} + \ket{\psi_1}$, where $\ket{\psi_0}$ is
its projection on the zero violations sector and $\ket{\psi_1}$ is
its projection on the violating sectors. Clearly $\Pi_m\ket{\psi} =
\ket{\psi_0}$. To
approximate this behavior, we will use the $\{P_i\}$ projections to
construct an operator $\hPi_m$ that is diagonal in the sectors
decomposition, and in addition $\hPi_m\ket{\psi} = \ket{\psi_0} +
\ket{\psi'_1}$, with $\ket{\psi'_1}$ in the violating sectors and
$\norm{\psi_1'}^2 \le \delta\norm{\psi_1}^2$. It follows that the
operator $\hA\EqDef\hPi_m\Pi_{rest}\Pi_{odd}$ approximates
$\Pi_{even} \Pi_{odd}$ in the sense that $\hA\ket{\gs}=\ket{\gs}$,
$\hA\ket{\gsp}\in\Hip$, and
$\norm{\hA\ket{\gsp}}^2\le(\Delta_0^k+\delta)\norm{\gsp}^2$. Let
$\hDel \EqDef \Delta_0^k + \delta$.

To construct the operator $\hPi_m$, first consider the operator $\BN
= \sum_{i=1}^m (\Id-P_i)$. The operator $\BN$ counts the number of
violations in a sector: If $\ket{\psi_{\Bs}}$ belongs to the $\Bs$
sector, $\BN\ket{\psi_{\Bs}} = |\Bs|\cdot\ket{\psi_{\Bs}}$. The operator
$\hPi_m$ will be a polynomial in $\BN$, with the polynomial
evaluating to $1$ on $|s|=0$, and less than $\delta$ on input with
$|s|$ between $1$ and $m$. Three ideas play a critical role in the
construction of this polynomial and in bounding the increase in
Schmidt rank. The first is the use of a Chebyshev polynomial, which
achieves the desired behavior at the $m+1$ points with a degree of
only $j = \orderof{\sqrt{m}\log\delta^{-1}}$. The second idea is
that it suffices to bound the entanglement across any of the $m$
cuts and then to pay a further penalty of, at most, $D_I\EqDef
(D_0^k)^m$ to bound the entanglement across the cut of interest. So,
if we consider the operator $\hA^\ell$, each term has degree $j\ell$
(i.e., is a product of $j\ell$ of the $P_i$s), and so the typical cut
is crossed $j\ell/m$ times, resulting in a Schmidt rank increase by
$(D_0^k)^{j\ell/m} \simeq D_0^{k\ell/\sqrt{m}}$. This means that the
incremental Schmidt rank per application of a term of $\hA$ is
$D_0^{k/\sqrt{m}}$, which can be made arbitrarily small by choosing
$m$ to be large enough. Finally, a recursive grouping argument shows that
we do not have to pay a price in Schmidt rank proportional to the
number of terms in the polynomial (which would have been
catastrophic); instead, we can decompose the operator $\hA^\ell$ as 
a sum of only $2^{\orderof{\log^2 j}}$ operators, for each of which
there is a (possibly different) cut with entanglement increase of
$\simeq D_0^{k\ell/\sqrt{m}}$. 

Putting it all together, we have an operator $K = \hA^\ell$, which 
increases the Schmidt rank by $D = D_I \hD^\ell$, where $\hD =
2^{\orderof{\log^2 j}}D_0^{k/\sqrt{m}}$, and
$j=\orderof{\sqrt{m}\log \delta^{-1}}$ and achieves a shrinkage
factor of $\Delta = \hDel^\ell$ for $\hDel = \Delta_0^k +
\delta$. It is now a matter of simple algebra to fix the parameters
$m,\delta, k$ and $\ell$ such that $D\cdot\Delta<1/2$. The end results
turns out to be $\log D = \orderof{1}\cdot X^3\cdot\log^8X$ for
$X=(\log D_0)/\epsilon$, and this completes the proof.

{~}

\noindent\textbf{Paper Organization:}\\

We begin with some preliminary definitions and known results from
mathematics and quantum information in \Sec{sec:Notation}. In
\Sec{sec:proof} we give the proof of our main result,
\Thm{thm:main}. The heart of the proof, which is its most technical
part, is the diluting lemma. It is proved separately in
\Sec{sec:diluting}. In \Sec{sec:2D} we provide an outline for our
entanglement bound in 2D and beyond (the ``almost volume law''
result), and in \Sec{sec:MPS} we sketch the implication of our area
law to the existence of a matrix product state approximation for the
ground state. In \Sec{sec:Conclusions} we offer our summary and
conclusions.

\section{Notation and Preliminaries}
\label{sec:Notation}

Throughout this paper $\log(\cdot)$ will denote the base 2
logarithmic function.

\subsection{Local Hamiltonians}

We consider a $k$-local Hamiltonian $H$ acting on
$\Hi=(\mathbbm{C}^d)^{\otimes n}$, the Hilbert space of $n$
particles (spins, qudits) of dimension $d$ that sit on a
$D$-dimensional grid. Our main result concerns the 1D case with
$D=1$, but in \Sec{sec:2D} we will consider higher dimensions. We
assume $H=\sum_i H_i$ where each $H_i$ is a non-negative bounded
operator that acts non-trivially on a constant number of $k$
particles (hence the term \emph{local Hamiltonian}). We further assume that
$H$ has a unique ground state $\ket{\gs}$ with ground energy $0$ and
has constant spectral gap $\epsilon>0$. Since the $H_i$ are all
non-negative, the ground state must be a common zero eigenstate for
each of the $H_i$; a Hamiltonian with this feature is known as
\emph{frustration free}. We denote by $\Hip \subset \Hi$ the
orthogonal complement of the ground space of $H$. Thus, $\Hip$ is an
invariant subspace for $H$ and 
\begin{align} 
\label{eq:gapcondition}
  H|_{\Hip} \geq \epsilon \Id \ . 
\end{align}

We further assume that the $H_i$ terms are projections, and,
henceforth, we will denote them by $Q_i$ to remind the reader. We
define $P_i$ to be the projection on the ground space of $Q_i$,
$P_i\EqDef\Id-Q_i$.  The assumption that $H$ is made of projections
is not actually a restriction, as is demonstrated, for example, in
Sec.~2 of \Ref{ref:DL2}. Indeed, given a general frustration-free
system with a spectral gap $\epsilon>0$ and interaction strength
bounded by $\norm{H_i}\le J$, one can always pass to an equivalent
system that shares the same ground space, which is made of
projections and has a spectral gap $\epsilon/J$. Therefore,
throughout the paper, we will drop the $J$ dependence; we will
assume $J=1$ and work in dimensionless units.

\subsection{The Schmidt decomposition}

Given a state $\ket{\phi}$ and a bipartition of the system to two 
non intersecting sets, $R$ and $L$, with corresponding Hilbert
spaces $\Hi_L, \Hi_R$ such that $\Hi=\Hi_L\otimes\Hi_R$, we can
consider the \emph{Schmidt decomposition} of the state along this
cut: $\ket{\phi} = \sum_j \lambda_j \ket{L_j}\otimes\ket{R_j}$. Here
$\lambda_1\ge \lambda_2\ge \ldots$ are the \emph{Schmidt
coefficients}. Their squares sum 1 (if $\ket{\phi}$ is normalized)
and are equal to the non-zero eigenvalues of the reduced density
matrices of either sides of the cut. 

The number of nonzero Schmidt coefficients in the Schmidt
decomposition of $\ket{\phi}$ is called the \emph{Schmidt rank} (SR),
which we shall denote as $\SR{\phi}$.  The usefulness of the SR stems
from it being a ``worst case'' estimate for the entanglement. As
such, it is often easy to bound.  The following facts are easy to
verify:
\begin{fact} \ 
\label{f:SR}
  \begin{enumerate}
    \item $\SR{\phi + \psi} \le \SR{\phi} + \SR{\psi}$.
    
    \item If $O$ is a $k$-local operator  whose support intersects both
      $\Hi_L$ and $\Hi_R$, then it can increase the SR (with respect
      to the bi-partitioning) of any state by, at most, a factor of
      $d^k$: $\SR{O\phi} \le d^k\SR{\phi}$. If $O$ intersects only
      one part of the system, its action cannot increase the SR.
      
    \item Consider a 1D system. If $r_i$ and $r_j$ are the SR of 
      $\ket{\phi}$ that correspond to cuts between particles $i,
      i+1$, and $j,j+1$, then $d^{-|i-j|}r_j\le r_i \le d^{|i-j|} r_j$.
  \end{enumerate}
\end{fact}

An important fact about the SR is the following corollary of the
Eckart-Young theorem \cite{ref:Young-Eckart36}, which states that
the truncated Schmidt decomposition provides the best approximation
to a vector in the following sense:
\begin{fact}  
\label{f:rankapprox} 
  Let $\ket{\psi}$ be a vector on a bi-partitioned Hilbert space
  $\Hi_L \otimes \Hi_R$, and let $\lambda_1\ge \lambda_2\ge\ldots$
  be its corresponding Schmidt coefficients.  Then the largest inner
  product between $\ket{\psi}$ and a normalized vector with Schmidt
  rank $r$ is $\sqrt{\sum_{j=1}^r\lambda^2_j}$.
\end{fact}

\subsection{The Detectability Lemma}
\label{sec:DL}
  
One of our main technical tools in this paper is the Detectability
Lemma (DL).  Originally proved in \Ref{ref:DL} in the context of
quantum constraint satisfaction and promise gap amplification, a
simpler and stronger version of the DL was proved in \Ref{ref:DL2},
where it was used in the context of gapped, frustration-free local
Hamiltonians. This is the form that would be used here. To state it,
consider  a gapped 1D frustration-free
Hamiltonian with nearest-neighbor interactions that is defined on
a chain of $n$ $d$-dimensional particles,
\begin{align}
  H=\sum_{i=1}^{n-1} Q_i \ . 
\end{align}
As explained previously, we assume that the 2-local interactions
terms $Q_i$ are projections, and we set $P_i\EqDef \Id-Q_i$ to be
the projection of the local ground space of every term. 

\begin{figure*}
  \begin{center}
    \includegraphics[scale=0.8]{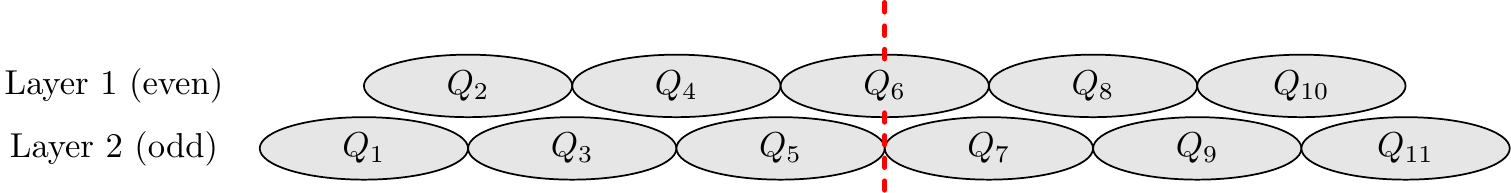}
  \end{center}  
  \caption{ The  settings  of the detectability lemma for a 1D system
  with 12 particles  $H=\sum_{i=1}^{11}   Q_i$. The local interaction
  terms, $Q_i$, are divided into two  layers:  the even layer and the
  odd layer.  The dashed red line denotes a cut in the system between
  particles 6 and 7. \label{fig:2layers} }
\end{figure*}  

We partition the projections into two sets which we call
\emph{layers}: the odd layer $Q_1, Q_3, Q_5, \ldots$ and the even
layer $Q_2, Q_4, Q_6, \ldots$. Within each layer the projections are
non-intersecting and, therefore, are commuting. It follows that
$\Pi_{odd} \EqDef P_1\cdot P_3\cdot P_5\cdots$ and $\Pi_{even}
\EqDef P_2\cdot P_4\cdot P_6\cdots$ are the projections into the
ground space of the odd and even layers respectively.  An
illustration of this construction is shown in \Fig{fig:2layers}. We
then have:

\begin{lem}[Detectability Lemma (DL) in $1D$]
\label{lem:DL}
  Let $A\EqDef\Pi_{odd}\Pi_{even}$, and let $\Hip$ be the orthogonal
  complement of the ground space. Then
  \begin{align}
  \label{eq:shrink}
      \norm{A|_{\Hip}}^2
        \leq \Delta_0(\epsilon)\EqDef\frac{1}{(\epsilon/2 + 1)^{2/3}} \ .
  \end{align}
\end{lem}
For brevity, we will often drop the $\epsilon$ dependence, and simply write
$\Delta_0$. We note that the DL is not restricted to the 1D case,
and can be easily generalized to any dimension by using more than two
layers.

In the commuting case, $\Pi_{even}$ and $\Pi_{odd}$ commute with
each other, and their product, $A\EqDef \Pi_{odd}\cdot\Pi_{even}$ is
the projection into the ground state of the system. Generally,
however, they do not commute, and as a result $A$ is only an
\emph{approximation} to the ground space projection: It leaves the
ground space invariant while shrinking its perpendicular space by
some factor. The DL quantifies this approximation: It tells us that
the shrinking factor is bounded away from 1 by a constant
$\Delta_0(\epsilon)$ that depends on the spectral gap and not on the
system size.

\subsection{Chebyshev polynomials of the first kind}
\label{sec:Cheby}

We recall some basic facts about Chebyshev polynomials of the first
kind. More information can be found in Chapter 22 of
\Ref{ref:AbramowitzStegun64} or in any standard text book on
approximation theory. This is a family of polynomials $\{T_n(x)\}$,
where $n$ denotes the degree of the polynomial. They are solutions
of the ODE
\begin{align}
\label{eq:defeq}
  (1-x^2)\frac{d^2y}{dx^2}-x\frac{dy}{dx}+n^2 y =0 \ .
\end{align}
For $|x|\le 1$ they are given explicitly by 
\begin{align}
  \label{def:Tn}
  T_n(x) \EqDef \cos(n\cos^{-1}x) \ ,
\end{align}
and, for general $x$, they are given recursively by
\begin{align}
  T_0(x) &= 1 \ ,\\
  T_1(x) &=x \ , \\
  T_{n+1}(x) &= 2xT_n(x)-T_{n-1}(x) \ .
\end{align}

By the recursion relation, the coefficient of $x^n$ in $T_n(x)$ is
$2^{n-1}$, and by \Eq{def:Tn}, it has $n$ simple roots (Chebyshev
nodes) in $(-1,1)$, $x_k =
\cos\left(\frac{\pi}{2}\frac{2k-1}{n}\right), k=1, \ldots, n$.
Therefore,
\begin{align}
\label{eq:roots}
  T_n(x) = 2^{n-1}\prod_{k=1}^n (x-x_k) \ .
\end{align}

We conclude this section with an inequality about the behavior
$T_n(x)$ near $x=-1$ that will be used in the proof of the diluting
lemma. From \Eq{def:Tn} it follows that for all $-1\le x\le 1$,
$|T_n(x)|\le 1$ and that $T_n(1) = 1$; $T_n(-1)=(-1)^n$. Finally,
let us look at the derivative of $T_n(x)$ at $x=-1$: substituting
$x=-1$ at \Eq{eq:defeq} gives us $T_n'(-1) + n^2T_n(-1)=0$ and so
$T_n'(-1)=-n^2T_n(-1) = (-1)^{n+1}n^2$.  Since $T_n(x)$ has all its
minimas/maximas inside $[-1,1]$, its derivative $T'_n(x)$ outside
that region can only grow in absolute value, and, therefore, for all
$\delta>0$:
\begin{align}
\label{eq:Tn-outside}
   |T_n(-1-\delta)| \ge |T_n(-1)| + \delta |T_n'(-1)|\ge 1 + n^2\delta \ .
\end{align}

\section{Proof of the main theorem}
\label{sec:proof}

Recall from the outline of the proof presented in the introduction,
that our goal is to construct an operator $K$ whose effect is to
rapidly increase the overlap with the ground state, while only 
slowly increasing the Schmidt rank. We formalize this property of
$K$ in the notion of a $(D, \Delta)$-AGSP below and then show that
if the trade-off between $D$ and $\Delta$ if favorable, i.e.,
$D\cdot\Delta < 1/2$, then we can show that there is a product state
that has large overlap with the ground state, which, in turn, leads
to a bound on the entanglement entropy of the ground state.  Once
this is established, we can move on to the central construction of
this work, which is performed in the \emph{diluting
lemma}~\ref{lem:diluting}. This is where an operator with the proper
trade-off of $D$ and $\Delta$ is constructed.

We begin with a quantitative definition of an operator that moves
any vector toward the ground state.
\begin{deff}
  \textbf{An Approximate Ground-Space Projection (AGSP)} \\
  \label{def:AGSP} 
  Consider a local Hamiltonian system $H=\sum_i H_i$ on a 1D chain,
  together with a cut between particles $i^*$ and $i^*+1$ that
  bi-partitions the system. We say that an operator $K$ is a $(D,
  \Delta)$-Approximate Ground Space Projection (with respect to the
  cut) if the following holds:
  \begin{itemize}
    \item \textbf{Ground space invariance:} for any ground state
      $\ket{\gs}$, $K\ket{\gs} = \ket{\gs}$.
      
    \item \textbf{Shrinking:} for any state $\ket{\gsp}\in \Hip$,
      also $K\ket{\gsp}\in\Hip$, and $\norm{K\ket{\gsp}}^2 \le
      \Delta$.
      
    \item \textbf{Entangling:} for any state $\ket{\phi}$,
      $\SR{K \ket{\phi}} \le D \cdot\SR{\phi}$, where
      $\SR{\cdot}$ is evaluated with respect to the bi-partitioning.
  \end{itemize}
  
  We refer to $D$ as the SR factor and $\Delta$ as the shrinking
  factor.
\end{deff}
We note that with this definition, the DL combined with
Fact~\ref{f:SR} implies that $A=\Pi_{even}\Pi_{odd}$ is a $(D_0,
\Delta_0)\EqDef(d^2, (1+\epsilon/2)^{-2/3})$-AGSP.  In its bare
form, however, this operator is not useful to us since its SR factor
is too large with respect to its shrinking factor. Specifically, it
turns that the important quantity to consider is the product
$D\cdot\Delta$.  The following lemma shows how having a
$(D,\Delta)$-AGSP with $D\cdot\Delta < \frac{1}{2}$ implies the existence
of a product state whose overlap with the ground state is at least
$1/\sqrt{2D}$.
\begin{lem}
\label{lem:mu1} 
  If there exists an $(D, \Delta)$-AGSP with $D\cdot\Delta \leq
  \frac{1}{2}$, then there is a product state
  $\ket{\phi}=\ket{L}\otimes\ket{R}$ whose overlap with
  the ground state is\footnote{Without loss of generality, we may absorb
  any phases of $\mu$ into $\ket{\phi}$ and assume $\mu$ is real and
  positive.} $\mu=\braket{\gs}{\phi} \ge 1/\sqrt{2D}$.
\end{lem}

\begin{proof}

  Let $K$ be a $(D, \Delta)$-AGSP with $D\cdot\Delta \leq \frac{1}{2}$
  and $\ket{\phi'}$ a product state
  $\ket{\phi'}\EqDef\ket{L'}\otimes\ket{R'}$ whose overlap with the
  ground state is $\mu=\braket{\gs}{\phi'} < 1/\sqrt{2D}$. Below we
  show that there exists another product state with a larger overlap
  with the ground state.

  By assumption, we can expand $\ket{\phi'}$ as $\ket{\phi'} =
  \mu\ket{\gs} + (1-\mu^2)^{1/2}\ket{\gs^\perp}$, with
  $\ket{\gs^\perp}\in \Hip$. Let $\ket{\phi_1}\EqDef K\ket{\phi'}$.
  Since $K$ is an $(D,\Delta)$-AGSP, it follows that $\ket{\phi_1} =
  \mu\ket{\gs} + \delta_1\ket{\gs^\perp_1}$ with
  $\ket{\gs^\perp_1}\in \Hip$, $\delta_1^2 \le \Delta$, and
  $\SR{\phi_1}\le D$.

  Setting $\ket{v}= \frac{1}{\norm{\phi_1}}\ket{\phi_1}$ to be the
  normalization of $\ket{\phi_1}$, we have $\SR{v}=\SR{\phi_1}\le
  D$. Therefore, its Schmidt decomposition can be written as
 $\ket{v} = \sum_{i=1}^D \lambda_i\ket{L_i}\ket{R_i}$.
  It follows that 
  \begin{align*}
   |\braket{\gs}{v}| &= \frac{\mu}{\norm{\phi_1}}
      \le \sum_{i=1}^D \lambda_i |\braket{\gs}{L_i}\ket{R_i}|
     \le \sqrt{\sum_{i=1}^D |\braket{\gs}{L_i}\ket{R_i}|^2} \ ,
  \end{align*}
  where the last inequality follows from Cauchy-Schwartz and the
  fact that $\sum_i\lambda_i^2=1$. Therefore there must be an $i$
  for which 
  \begin{align*}
    |\braket{\gs}{L_i}\ket{R_i}|^2 
     \ge \frac{\mu^2}{D\norm{\phi_1}^2}
     = \frac{\mu^2}{D(\mu^2+\delta_1^2)}
     \ge \frac{\mu^2}{D(\mu^2+\Delta)} \ .
  \end{align*}
  But since $D\cdot\Delta<1/2$, and, by assumption
  $\mu<\frac{1}{\sqrt{2D}}$, it follows that $D(\mu^2+\Delta)<1$,
  and so the overlap of $\ket{L_i}\ket{R_i}$ with the ground state
  is larger than $\mu$. 
\end{proof}

With this bound in place, we start from the product state with the
maximal overlap with the ground state, and use any AGSP to obtain
controlled approximations of the ground state, from which an upper
bound on its entropy can be found. A very similar argument was used
in Hastings' proof of the 1D area law.\cite{ref:Has07}
\begin{lem}
\label{lem:S} 
  If there exists a product state whose overlap with the ground
  state is at least $\mu$, together with a $(D, \Delta)$-AGSP,
  then the entanglement entropy of $\ket{\gs}$ is bounded by
  \begin{align}
  \label{eq:Sb}
    S\le \orderof{1}\cdot\frac{ \log \mu^{-1}}{\log \Delta^{-1}} \log D \ .
  \end{align}
 \end{lem}

%
%

The proof can be found in the appendix.  The brief overview is that we
begin with the asserted product state, and repeatedly apply the
AGSP to it.  The result is a series of vectors with increasing SR
by a factor $D$ each time, that approach the ground state at a
rate quantified by powers of $\Delta$.  Using these vectors and the
Young-Eckart theorem (Fact~\ref{f:rankapprox}) provides adequate
upper bounds for the high Schmidt coefficients of the ground state
to bound the entropy. 

Lemmas~\ref{lem:S}~and~\ref{lem:mu1} can be combined to give
\begin{corol}
  \label{col:half}
  If there exists an $(D,\Delta)$-AGSP such that
  $D\cdot\Delta\le\frac{1}{2}$, the ground state entropy is bounded
  by:
  \begin{align}
  \label{eq:AL}
    S\le \orderof{1}\cdot\log D \ .
  \end{align} 
 \end{corol}

We are left, therefore, with the challenge of designing an operator
$K$ which is a $(D,\Delta)$-AGSP with $D\cdot\Delta\le 1/2$. This is
the driving construction of this work. It is stated in the
following lemma, and is proved in the next section. 

\begin{lem}[The diluting lemma]
\label{lem:diluting} Consider a 1D gapped frustration-free Hamiltonian,
  with a spectral gap $\epsilon>0$ and particle dimension $d$, and
  define $X\EqDef \frac{\log d}{\epsilon}$. Then for any cut in the
  chain there exists an $(D,\Delta)$-AGSP with $D\cdot\Delta<1/2$ and
  \begin{align}
  \label{eq:D-bound}
    \log D \le\orderof{1}\cdot X^3\log^8 X \ .
  \end{align}
\end{lem}
Substituting the result of this lemma in Corollary~\ref{col:half}
proves \Thm{thm:main}.

\section{Proving the diluting lemma (\Lem{lem:diluting})}
\label{sec:diluting}

We will prove the diluting lemma, by modifying the DL operator
$A\EqDef\Pi_{even}\Pi_{odd}$ to a new operator $\hA$, which is
\emph{not} an AGSP, but has similar properties (see
\Def{def:AGSP})
\begin{itemize}
  \item \textbf{Ground space invariance:} for any ground state
    $\ket{\gs}$, $\hA\ket{\gs} = \ket{\gs}$.
    
  \item \textbf{Shrinking:} for any state $\ket{\gsp}\in \Hip$,
    also $\hA\ket{\gsp}\in\Hip$, and $\norm{\hA\ket{\gsp}}^2 \le \hDel$.
    
  \item \textbf{Entangling:} for any state $\ket{\phi}$, and any
  integer $\ell>0$,
    $\SR{\hA^\ell \ket{\phi}} \le D_I\hD^\ell \cdot\SR{\phi}$.
    
  \item The parameters $D_I, \hD, \hDel$ satisfy
    \begin{align}
      \hD\cdot\hDel &\le \frac{1}{2} \ , \\
      \log \hD  &\le \log(\hDel^{-1}) \le \orderof{1}\cdot\log^2X \ , \\
     \log D_I   &\le \orderof{1} \cdot X^3\log^6 X \ .
    \end{align}
  
\end{itemize}

Having $\hA$ in hand, we can choose $\ell_0 = \lceil\log D_I\rceil$
and obtain the $(D,\Delta)$-AGSP 
\begin{align}
  \label{def:K}
  K\EqDef \hA^{\ell_0} 
\end{align}
with $(D,\Delta)=(D_I\hD^{\ell_0}, \hDel^{\ell_0})$. It satisfies 
$D\cdot\Delta<1/2$ and
\begin{align*}
  \log D =\orderof{1}\cdot \log D_I\cdot\log \hD
    \le\orderof{1}\cdot X^3\log^8 X \ .
\end{align*}

We now proceed to define the general form of $\hA$.

\subsection{General settings}

Without loss of generality, we assume that the bi-partitioning cut
in the chain intersects with an even projection (see
\Fig{fig:2layers}).  As a result, when applying $A$, only the
$\Pi_{even}$ portion of the operator increases the SR. Therefore,
our construction of $\hA$ will modify $\Pi_{even}$, leaving
$\Pi_{odd}$ intact.

\begin{figure*}
 \begin{center}
   \includegraphics[scale=1]{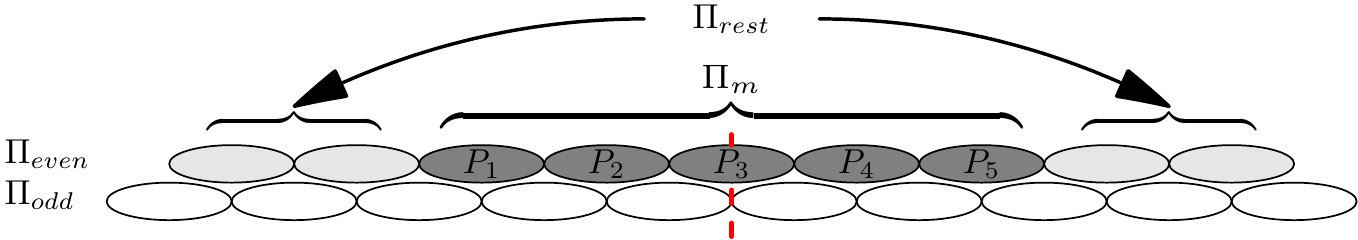}
 \end{center}
 \caption{An illustration of the decomposition
 $A=\Pi_{even}\Pi_{odd} = \Pi_{rest}\Pi_m\Pi_{odd}$ that is used to
 define $\Pi_m$ and its dilution $\hPi_m$. $\Pi_m$ is the product of
 the $m$ projections $P_1, \ldots, P_m$ that are found in the even
 layer around the cut.
 \label{fig:Im} }
\end{figure*}
We begin by considering the piece of $\Pi_{even}$ consisting of the
$m$ even projections closest to the cut.  We denote this set of
projections by $I_m$, and abusing the notation a
little, we relabel them by $P_1, \ldots, P_m$ and define
\begin{align}
  \Pi_m \EqDef \prod_{i=1}^m P_i \ ,
\end{align}
the projection into the common ground state of these projections. Note
that in this notation the cut intersects with $P_{i^*}$, where
$i^*=\lceil m/2 \rceil$ (see \Fig{fig:Im}).  Denoting by $\Pi_{rest}$
the product of all the remaining projections in $\Pi_{even}$, we have
$\Pi_{even}=\Pi_m\cdot\Pi_{rest}$. We will approximate $\Pi_m$ by an
operator $\hPi_m$, and then define $\hA$ by
\begin{align}
  \label{def:hA}
  \hA\EqDef \hPi_m\cdot\Pi_{rest}\cdot\Pi_{odd} \ .
\end{align}
Analysis of the amount of entanglement created by powers of $\hA$
will focus exclusively on the structure of $\hPi_m$ since
$\Pi_{rest}, \Pi_{odd}$ do not increase the SR along the cut.

One of the great benefits of using the DL is that the all
projections in a given layer commute, and hence the analysis becomes
almost classical.  Indeed, as we said in the outline of the proof, 
the projections in $I_m$ define a decomposition of
the Hilbert space of the system into a direct sum of $2^m$
eigenspaces, called sectors. Each sector is defined by a string
$\Bs=(s_1, \ldots, s_m)$, such that if $\ket{\psi_{\Bs}}$ is in the
$\Bs$ sector, $P_i\ket{\psi_{\Bs}} = (1-s_i)\ket{\psi_{\Bs}}$. A
site with $s_i=1$ is called \emph{a violation}, since it corresponds
to a non-zero energy of the corresponding local Hamiltonian term. We
denote by $|\Bs|=\sum_{i=1}^m s_i$ the total number of violations in
the sector $\Bs$.  Finally, we will also consider a coarse-grained
decomposition in which we group together all $\binom{m}{k}$ sectors
with $k$ violations. The direct sum of these subspaces is called the
$k$-violations sector.

Consider now an arbitrary state $\ket{\psi}$ and expand it as
$\ket{\psi}=\ket{\psi_0} + \ket{\psi_1}$, where $\ket{\psi_0}$ is its
projection on the zero violations sector and $\ket{\psi_1}$ is its
projection on the violating sectors. Clearly, $\Pi_m\ket{\psi} =
\ket{\psi_0}$.  To approximate this behavior, we will use the
$\{P_i\}$ projections to construct an operator $\hPi_m$ that is
diagonal in the sectors decomposition, and in addition
$\hPi_m\ket{\psi} = \ket{\psi_0} + \ket{\psi'_1}$, with
$\ket{\psi'_1}$ in the violating sectors and $\norm{\psi_1'}^2 \le
\delta\norm{\psi_1}^2$ for some error parameter $\delta>0$. Then it is
easy to verify that
$\hA\ket{\gs}=\ket{\gs}$, $\hA\ket{\gsp}\in\Hip$, and
$\norm{\hA\ket{\gsp}}^2\le(\Delta+\delta)\norm{\gsp}^2$. This gives us
\begin{align}
\label{def:hDel}
  \hDel = \Delta_0 + \delta \ .
\end{align}

\subsection{Constructing $\hPi_m$ -- a general discussion}

Before actually constructing the operator $\hPi_m$, it might be
useful to describe what we hope to accomplish. Recall that we would
like to show that the increase in SR due to application of
$\hA^\ell$ is bounded by $D_I \hD^{\ell}$, where $\hD\cdot\hDel \le
\frac{1}{2}$.  The rough idea is to show that there is some cut
between particles $i$ and $i+1$ within the support of $I_m$ such
that the SR across this cut does not grow much due to application of
$\hA^\ell$. This would account for the factor $\hD^{\ell}$ in the
above bound. Moreover, by Fact~\ref{f:SR}, the SR across the middle
cut $i^*,i^*+1$ can only be larger than the SR between $i, i+1$ by a
factor of at most $d^{|i-i^*|} \le  D_0^m \EqDef D_I$. 

There are several approaches to constructing $\hPi_m$. Perhaps the most
obvious one is by ``diluting''; instead of using a product of $m$
projections, we may use a product of $rm$ randomly chosen
projections for some $0<r<1$. After applying such $\hPi_m$ for
$\ell$ layers, there would be columns with fewer than $r\ell$
entangling projections, and so the SR along these columns will be
bounded by $D_0^{r\ell}$. Applying Fact~\ref{f:SR} as described 
above, the SR in the middle cut would be at most
$D_0^{r\ell}D_0^m$, implying $\hD=D_0^r$ and $D_I = D_0^m$.

What is the $\hDel$ factor of such construction? Intuitively, sectors
with high number of violations are easier to ``catch'' because there
is a higher chance of collision between one of their violations and
the $rm$ projections. Indeed, on average, the mass in the $k$ sector
is shrunk by at least
$\binom{m-k}{rm}/\binom{m}{rm}\le(1-r)^k$. But this means
that the low-violations sectors, and in particular the one-violation
sector, are barely shrunk. The latter can be shrunk by as little
as $1-r$, which means that $\hDel = \Delta_0 + 1-r$. This can never
get us to $\hD\cdot\hDel<1/2$.

To overcome this problem, we take a different approach for the
construction of $\hPi_m$ using the following operator:
\begin{deff}[The $\BN$ operator]
\label{def:BN}
  \begin{align}
    \label{eq:BN}
    \BN \EqDef \sum_{i=1}^m (\Id-P_i) \ .
  \end{align}
\end{deff}
The operator $\BN$ counts the number of violations in a sector: if
$\ket{\psi_{\Bs}}$ belongs to the $\Bs$ sector, $\BN\ket{\psi_{\Bs}}
= |\Bs|\cdot\ket{\psi_{\Bs}}$.

We can use the $\BN$ operator to annihilate the mass in the
low-violations sectors. Consider, for example, the operator $\BB_k
\EqDef \Id-\frac{1}{k}\BN$. When acting on $\ket{\psi_{\Bs}}$ we get
$\norm{\BB_k\ket{\psi_{\Bs}}}^2 =
\left(1-\frac{|\Bs|}{k}\right)^2\norm{\psi_{\Bs}}^2$. Therefore
$\BB_k$ completely annihilates the mass in the $k$ sector, while
leaving the 0 violations sector intact. In addition, as we shall see
in \Sec{sec:SR}, it does not create much entanglement.
Unfortunately, however, it blows up the mass of the high-violating
sectors $|\Bs|>k$. One possible solution is therefore to use the
$\BB_k$ operators in conjunction with the diluting approach. This
was indeed the approach taken in a previous paper
\cite{ref:FOCS2011-AreaLaw}, and it might prove beneficial in other
contexts. Here, however, we use a much simpler solution that relies
solely on the $\BN$ operators by utilizing the Chebyshev
polynomials.

\subsection{Constructing $\hPi_m$ using Chebyshev polynomials}

To construct $\hPi_m$ entirely from $\BN$, we want a polynomial of
minimal degree $P(x)$ such that $P(0)=1$ and $|P(x)|^2 \le \delta<1$
for every integer $1\le x\le m$. In such case, $\hPi_m \EqDef
P(\BN)$ will leave the zero-violations sector intact, while
shrinking the other sectors by $\delta$. One naive approach is to
take $P(x)=(1-x)\cdot(1-x/2)\cdots(1-x/m)$. This gives $\delta=0$ at
the price of a polynomial of degree $m$, the same degree as the
original $\Pi_m$ (thereby creating too much entanglement).  For a
lower degree polynomial with the desired properties we turn to the
Chebyshev polynomial, a central object in approximation theory.  As
noted in \Sec{sec:Cheby}, the Chebyshev polynomial of the first kind
$T_n(x)$ oscillates between $-1$ and $1$ in the region $[-1,1]$ and
then increases rapidly outside that region. The idea is therefore to
use a polynomial that is the mapping of the Chebyshev polynomial
from $[-1,1]$ to $[1,m]$ and rescaled to be 1 at $x=0$:
\begin{deff}[The $C_m(x)$ polynomial]
  The $C_m(x)$ polynomial is a $\sqrt{m}$-order polynomial that is 
  defined by
  \begin{align}
  \label{def:Ch}
  \hat{C}_m(x) &\EqDef
    T_{\sqrt{m}}
      \left(\frac{x-\frac{m+1}{2}}{\frac{m-1}{2}}\right) \ , \\
    C_m(x) &\EqDef \frac{1}{\hat{C}_m(0)} \hat{C}_m(x) \ .
    \label{def:C}
\end{align}
\end{deff}
Figure~\ref{fig:Tn} shows $C_m(x)$ for $m=36$.
It is easy to verify that 1) $C_m(0)=1$, and 2)
$|C_m(x)|^2\le\frac{1}{9}$ for every $1\le x\le m$. The first claim
follows from definition, while the second follows from the fact that
$|\hat{C}_m(0)|\ge 3$, which follows from \Eq{eq:Tn-outside}:
\begin{align}
  |\hat{C}_m(0)| 
      &=\left|T_{\sqrt{m}}
      \left(-1-\frac{2}{m-1}\right)\right| \\
      & \ge 1 + \frac{2}{m-1}(\sqrt{m})^2
      \ge 3\ . \nonumber
\end{align}
\begin{figure*}
  \begin{center}
    \includegraphics[scale=0.8]{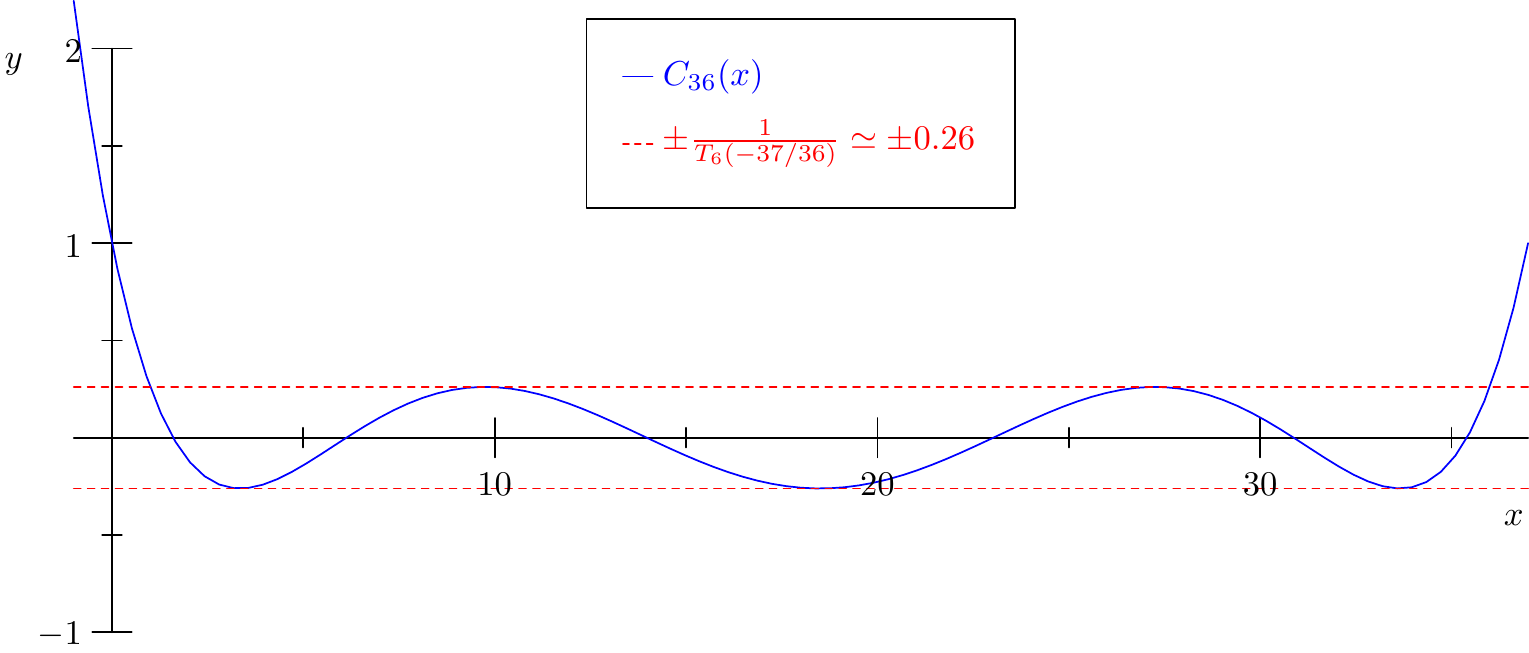}
  \end{center}  
    \caption{
    The polynomial $C_{36}(x)$, which is based on the Chebyshev
    polynomial of the first kind $T_6(x)$. For $x\in [1,36]$ we have
    $|C_{36}(x)|\le 1/3$.
    \label{fig:Tn} }
\end{figure*}

Notice that $C_m(x)$ is a polynomial of degree $\sqrt{m}$ -- a huge
improvement over the naive construction of degree $m$. In fact, it
can be shown that this construction is optimal: as shown in
\Ref{ref:poly-deg}, any polynomial that satisfies the above two
properties must be of at least degree $\sqrt{m}$.

To improve the shrinking factor, we can apply $[C_m(x)]^q$, where
$q$ is a parameter to be determined later.  We can now define the
operator $\hA$ that appears in the diluting lemma:
\begin{deff}[The Chebyshev-based operator $\hA$]
  Given integers $m$ and $q$, the Chebyshev-based operator $\hA$ is
  \begin{align}
    \label{def:KC}
    \hA \EqDef \hPi_m\cdot\Pi_{rest}\cdot\Pi_{odd} \ ,
  \end{align}
  where
  \begin{align}
  \label{def:hPm}
    \hPi_m \EqDef [C_m(\BN)]^q \ .
  \end{align}
\end{deff}
We note that $\hPi_m$ is a degree $j\EqDef q\sqrt{m}$ polynomial in
$\BN$, which leaves the zero-violations sector intact and shrinks
the violating sectors by at least $\delta=(1/9)^q$. Consequently, by
\Eq{def:hDel}, we get $\hDel = \Delta_0 + (1/9)^q$.  

Let us now turn to the task of upper bounding the SR that
is generated by $\hA^{\ell}$.

\subsection{Upper bounding the SR factor of $\hA^{\ell}$.}
\label{sec:SR}

We bound SR generated by $\hA^{\ell}$ in the following lemma:
\begin{lem}
\label{lem:mainSR}
  Set $j\EqDef q\sqrt{m}$, the degree of $[C_m(x)]^q$. Then for any
  state $\ket{\phi}$,
  \begin{align}
    \SR{\hA^\ell\phi} \le D_I \hD^\ell \cdot\SR{\phi} \ ,
  \end{align}
  where
  \begin{align}
    D_I &= D_0^m \ , \\
    \hD &= 20\cdot 2^{\frac{3}{2}\log j + \frac{1}{2}\log^2j}\cdot D_0^{j/m}
    \ .
  \end{align}
\end{lem}

\begin{proof}

Recall that $\hPi_m$ is a polynomial of degree $j$ in $\BN$ and
therefore can be written as a sum of $j+1$ terms:
\begin{align}
  \hPi_m = \sum_{i=0}^j c_i \BN^i \ .
\end{align}
Consequently, $\hA^\ell$ can be written as a superposition of $(j+1)^\ell$
terms of the form
\begin{align*}
 \BN^{i_\ell}\cdot(\Pi_{rest}\Pi_{odd})\cdots
 \BN^{i_2}\cdot(\Pi_{rest}\Pi_{odd})\cdot
 \BN^{i_1}\cdot(\Pi_{rest}\Pi_{odd}) \ ,
\end{align*}
with $i_1, \ldots, i_\ell$ between $0$ and $j$. When applied to a
general state $\ket{\phi}$, the term that potentially generates the
highest SR is $(\BN^j\cdot\Pi_{rest}\Pi_{odd})^\ell$. We will
therefore upper bound the total SR by upper bounding its SR and
multiplying the end result by $(j+1)^\ell$. 

\begin{deff}[A min-entangling operator]
  We say that an operator $C$ is min-entangling with respect to some
  cut in the support of $I_m$ if it is of the form
  \begin{align*}
     M_{i_\ell}\cdot(\Pi_{rest}\Pi_{odd})\cdots
     M_{i_2}\cdot(\Pi_{rest}\Pi_{odd})\cdot
     M_{i_1}\cdot(\Pi_{rest}\Pi_{odd}) \ ,
  \end{align*}
  where $M_i$ are polynomials in the projections of $I_m$ such that
  \begin{enumerate}
    \item There exists a subset of at most $j\ell/m$ of the $M_i$'s 
      that contains the projection $P \in I_m$
      that intersects with the cut.
      
    \item Each of the rest of the $M_i$ only contains projections that are
      either strictly to the left of $P$ or strictly to the right of
      $P$ (in $I_m$). 
 \end{enumerate}
\end{deff}

It follows that only those $M_i$ that contain the ``entangling''
projection $P$ increase the SR
across the cut, and therefore the total SR increase is bounded by
$D_0^{j\ell/m}$.  The bound on the SR follows from the following
decomposition of $\hA^{\ell}$:

\begin{claim}
\label{cl:mintangling}
 $(\BN^j\cdot\Pi_{rest}\Pi_{odd})^\ell$ can be written as a sum of
 at most $[4(j+1)\cdot (j/2+1)\cdot (j/4+1) \cdots (1+1)]^\ell$
 min-entangling operators. 
\end{claim}

\begin{proof}
  Say that an operator $C$ is $t$-min-entangling if there is a
  contiguous interval $I \subseteq I_m$ of $m/2^t$ projections such
  that $C$ is of the form 
  \begin{align*}
     M_{i_\ell}\cdot(\Pi_{rest}\Pi_{odd})\cdots
     M_{i_2}\cdot(\Pi_{rest}\Pi_{odd})\cdot
     M_{i_1}\cdot(\Pi_{rest}\Pi_{odd}) \ ,
    \end{align*}
  where a subset of at most $j\ell/2^t$ of the $M_i$'s are equal to
  $\BN_I\EqDef\sum_{i \in I} \Id-P_i$, while each of the other
  $M_i$'s is made of projections that are either to the left of $I$
  or to the right of $I$ (in $I_m$), and in particular does not
  include any projection from $I$.

  The proof relies on a recursive construction that after $t$ rounds
  decomposes $(\BN^j\cdot\Pi_{rest}\Pi_{odd})^\ell$ into a sum of
  $t$-min-entangling operators. After $\log m$ rounds of recursion
  we end up with the desired decomposition into min-entangling
  operators. 

  We begin by splitting $\BN$ into two terms: $\BN=\BN_L + \BN_R$,
  where $\BN_L\EqDef\sum_{i=1}^{m/2}(\Id-P_i)$ and
  $\BN_R\EqDef\sum_{i=m/2+1}^m (\Id-P_i)$.  Then $\BN^j = (\BN_L +
  \BN_R)^j = \sum_{i=0}^j \binom{j}{i}\BN_L^{j-i} \BN_R^i$.
  Expanding across all $\ell$ layers, we end up with $(j+1)^\ell$
  terms. Let us focus on one such term, and assume that at layer $i$
  it has the powers $\BN_L^{j-n_i}\BN_R^{n_i}$. Since the total
  degree in $\BN_L$ and $\BN_R$ across the $\ell$ layers is $j\ell$,
  it follows that the degree of one of the two across the $\ell$
  layers must be at most $j\ell/2$, and hence every term is
  1-min-entangling operator.  Consequently, $(\BN^j
  \Pi_{rest}\Pi_{odd})^\ell$ can be written as a sum of at
  $(j+1)^\ell$ 1-min-entangling operators.

  Proceeding recursively, we now pick one such term. It either
  contains at most $j\ell/2$ $\BN_L$ operators or at most $j\ell/2$
  $\BN_R$ operators. Assume without loss of generality it is
  $\BN_L$. We write $\BN_L = \BN_{LL} + \BN_{LR}$, and consequently,
  every $\BN_L^{n_i}$ can be written as the sum of $n_i+1$ terms:
  $\BN_L^{n_i} = \BN_{LL}^{n_i} + n_i\BN_{LL}^{n_i-1}\BN_{LR} +
  \ldots + \BN_{LR}^{n_i}$. Therefore, upon opening the product, we
  obtain a sum of $\prod_{i=1}^\ell (n_i+1)$ terms such that each
  term has either a maximal degree of $j\ell/4$ of $\BN_{LL}$ or of
  $\BN_{LR}$. These are all 2-min-entangling operators. Moreover,
  it is now easy to verify \emph{subject to the constraint}
  $\sum_{i=1}^\ell n_i \le j\ell/2$, the total number of terms
  $\prod_{i=1}^\ell (n_i+1)$ is maximized when all $n_i$ are equal:
  $\prod_{i=1}^\ell (n_i+1) \le (j/2+1)^\ell$. To summarize, we have
  just shown that $(\BN^j\Pi_{rest}\Pi_{odd})^\ell$ can be written
  as a sum of at most $(j+1)(j/2+1)$ 2-min-entangling operators.
  
  Continuing in this fashion for $\log j$ rounds we end up with a
  total of at most $[(j+1)\cdot (j/2+1)\cdot (j/4+1) \cdots
  (1+1)]^\ell$ $\log j$-entangling operators.  At this stage, the
  total degree of each such operator in $\BN_I$ is at most
  $j\ell/2^{\log j} = \ell$.  To bound the increase in the remaining
  $\log (m/j)$ rounds, we observe that subject to the constraint
  $\sum_{i=1}^\ell n_i \le \ell/2^k$ the expression
  $\prod_{i=1}^\ell (n_i+1)$ is maximized when $\ell/2^k$ of the
  $n_i$'s are $1$ and the rest $0$: $\prod_{i=1}^\ell (n_i+1) \leq
  2^{\ell/2^k}$. It follows that the total increase in the number of
  terms over these rounds is bounded by $2^\ell 2^{\ell/2}
  2^{\ell/4} \cdots \leq 4^\ell$.  This completes the proof of the
  claim. 
\end{proof}

We can now finish off the proof of the main lemma of this section.
Claim~\ref{cl:mintangling} gives a decomposition of the operator
$(\BN^j\cdot\Pi_{rest}\Pi_{odd})^\ell$ as a sum of at most
$[4(j+1)\cdot (j/2+1)\cdot (j/4+1) \cdots (1+1)]^\ell$
min-entangling operators. As noted above, to apply this result to
$\hA^\ell$ we must further multiply this number by $(j+1)^\ell$, so
$\hA^\ell$ can be written as a sum of no more than $[4(j+1)^2\cdot
(j/2+1)\cdot (j/4+1) \cdots (1+1)]^\ell$ terms. It is easy to verify
that for $j\ge 2$ (which is always the case), 
  \begin{align*}
    &4(j+1)^2\cdot (j/2+1)\cdot (j/4+1) \cdots
    (1+1) \\
    &\le 20 j^{3/2}2^{\frac{1}{2}\log^2 j} \ .
  \end{align*}

  The SR contribution of each such min-entangling operator \emph{at
  that cut that passes through its diluted column} is at most
  $D_0^{j\ell/m}$, and since this column is at most $m$ particles
  away from the bi-partitioning cut, it follows from Fact~\ref{f:SR}
  that its SR contribution to the bi-partitioning cut is at most
  $D_0^{j\ell/m}\cdot D_0^m$. Therefore $\SR{\hA^\ell \ket{\phi}}
  \le D_I\hD^\ell\cdot\SR{\phi}$, with 
  \begin{align}
  \label{eq:pre-DI}
     D_I &= D_0^m \ , \\
     \hD &= 20 j^{3/2}2^{\frac{1}{2}\log^2 j}D_0^{j/m} \ .
     \label{eq:pre-D}
  \end{align}
  This concludes the proof of \Lem{lem:mainSR}.
\end{proof}

At this point it is worthwhile to pause and take inventory of where
we are. We have shown that the operator $\hA^\ell$ is an AGSP with
characteristic factors $\{D_I \hD^{\ell}, (\Delta _0 + 
(1/9)^q)^{\ell}\}$.  We are searching for an AGSP whose product of
characteristic factors is below $1/2$.  This will be the case in our
situation for some $\ell$ as long as the product
$\hD\cdot[\Delta_0+(1/9)^q]$ is less than one.  Our work so far has
shown
\begin{align}
  \hD\cdot[\Delta_0 +(1/9)^q]=20j^{3/2}2^{\frac{1}{2}\log^2 j}D_0^{j/m} 
   \big[\Delta_0 + (1/9)^q\big] \ .
\end{align}
Recalling that $j=q\sqrt{m}$ it is clear that we can find constants
$q,m$ (in $D_0$) such that $\hD\cdot(1/9)^q < 1/2$; however the term
$\hD\cdot\Delta_0$ may still be bigger than $1$. 

To solve this problem another idea is needed: coarse-graining. As we
shall see in the next subsection, fusing together $k$ adjacent
particles allows us to move to a new local Hamiltonian system with
$D_0, \Delta_0$ replaced by $D_0^k, \Delta_0^k$. Taking $k
=\orderof{q\log \Delta^{-1}_0}$, would then yield $\Delta_0^k \simeq
(1/9)^q$, and consequently lead to $\hD\cdot\hDel < 1/2$.

\subsection{Coarse-graining}

\begin{figure}
  \begin{center}
    \includegraphics[scale=0.8]{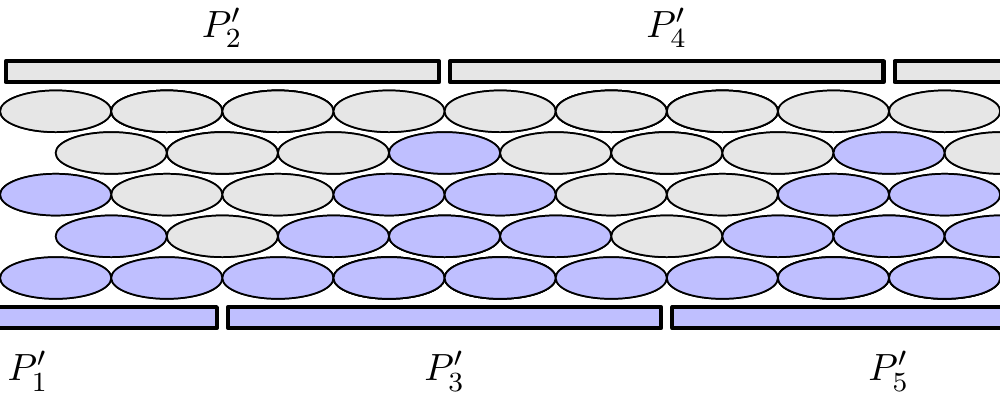}
  \end{center}  
    \caption{
    An illustration of a $k$-coarse grained system in 1D with $k=4$.
    The elongated rectangles denote the coarse-grained projections
    and the ovals denote the original projections. Underneath even
    coarse-grained projections, one can ``pull'' a pyramid of
    original projections, and similarly above an odd coarse-grained
    projection. Together, they form $k$ layers of the original
    projections. This shows that the coarse-grained shrinking
    exponent is actually $\Delta_0' = \Delta_0^k$.    
    \label{fig:coarse-grain} }
\end{figure}

Consider a $k$-coarse-grained system, in which we fuse together $k$
adjacent particles, making them a single particle of dimension
$d^k$. The new Hamiltonian of the system would now be a 2-local
Hamiltonian on a chain, consisting of projections $Q'_i = \Id-P'_i$,
where $P'_i$ denotes the projection into the common ground space of
the $2k$ particles that form the coarse grained particles $i$ and
$i+1$. We define the odd/even layers projections $\Pi'_{odd},
\Pi'_{even}$ accordingly, and notice that every application of
$\Pi'_{even}\Pi'_{odd}$ increases the SR by at most a factor of
$D_0^{k}$. To estimate the shrinking factor, we may use the DL on
the new Hamiltonian using the new spectral gap.

However, there is a much stronger bound that we can obtain by using
the DL on the \emph{original system}. Based on an idea 
that has already appeared in a related form in \Ref{ref:DL2},
we show that:
\begin{claim}
  \begin{align*}
    &\Pi'_{odd}\Pi'_{even} \\
    &= \Pi'_{odd} \cdot
    (\Pi_{even}\Pi_{odd}\Pi_{odd}\Pi_{even})^{k/2} \cdot \Pi'_{even} \ .  
  \end{align*}
\end{claim}
This immediately implies:

\begin{corol} 
\label{cor:rc-coarse} 
  If $H'=\sum_i H'_i$ is the $k$-coarse-grained version of $H=\sum_i
  H_i$, then the SR and shrinking factors of the DL AGSP of $H'$ are
  related to those of $H$ by
  \begin{align}
    (D'_0, \Delta'_0) = (D_0^k, \Delta_0^k) \ .
  \end{align}  
\end{corol}

The claim follows from two observations:
\begin{enumerate}
  \item For any coarse-grained constraint $P'_i$, we can always ``pull'' 
    from it a product of the original projections $P_i$ that act on
    the support of $P'_i$. The reason being that $P'_i$ projects
    into the common ground space of all the original $2k$ particles
    in its support, and the original projections $P_i$ are trivial
    on that space.

  In particular, let $Pyr_i$ denote the ``pyramid'' or
  ``light-cone'' of original projections pulled from $P'_i$ as in
  the figure above. Then $\Pi'_{odd} =
  \Pi'_{odd}\cdot\Pyr_1\cdot\Pyr_3\cdot\Pyr_5\cdots$, and
  $\Pi'_{even} = \Pyr_2\cdot\Pyr_4\cdot\Pyr_6\cdots\Pi'_{even}$.
  
  \item
    $\Pyr_1\cdot\Pyr_3\cdot\Pyr_5\cdots\Pyr_2\cdot\Pyr_4\cdot\Pyr_6\cdots =
    (\Pi_{odd}\Pi_{even})^{k+1}$.
    
  This is due to their pyramid-like shape, they can be commuted past
  each other and re-arranged as $k+1$ layers of the original
  projections (see the figure above). We note that applying $k+1$
  original layers of the DL shrinks\footnote{At first sight it seems
  that the shrinking should be of $\Delta_0^{(k+1)/2}$, but actually
  we can do much better by ``duplicating'' the projections in the
  middle. For example, for $k=3$, we get
  $\Pi_{even}\Pi_{odd}\Pi_{even}\Pi_{odd} =
  (\Pi_{even}\Pi_{odd})(\Pi_{odd}\Pi_{even})(\Pi_{even}\Pi_{odd})$.
  Every bracket contributes a $\Delta_0$ factor, so overall we get
  $\Delta_0^3$.} the perpendicular space by a factor of
  $\Delta_0^k$.
\end{enumerate}

\subsection{Gluing  it all together}

If we first $k$-coarse-grain the system, and then
construct $\hA$, we obtain the following factors:
\begin{align}
\label{eq:free-DI}
  D_I &= D_0^{km} \ , \\
  \hD &= j^{3/2}2^{\frac{1}{2}\log^2 j}D_0^{kj/m} \ , \label{eq:free-D}\\
  \hDel &= \Delta_0^k + (1/9)^q \ .
  \label{eq:free-Delta}
\end{align}
We have 3 free parameters: $m, q, k$ (recall that $j=q\sqrt{m}$). To
finish the proof we show how these can be chosen to obtain
$\hD\cdot\hDel<1/2$.

Our first step is to demand that $k$ is large enough such that
$\Delta_0^k\le(1/8)^q$. Looking at \Eq{eq:shrink}, it is easy to
verify that as long as $\epsilon\le 10$, 
this can achieved by defining, for example, 
\begin{align}
  \label{eq:fix-k}
  k\EqDef \frac{20q}{\epsilon} \ .
\end{align}
Then $\hDel \le 2(1/8)^q$.  A sufficient condition for
$\hD\cdot\hDel<1/2$ is therefore
\begin{align}
  \label{eq:main}
  j^{3/2}2^{\frac{1}{2}\log^2j}
   D_0^{\frac{1}{\epsilon}20qj/m}2^{-3q} \le \frac{1}{80} \ , 
\end{align}
or equivalently, 
\begin{align}
\label{eq:1d-log}
  \frac{3}{2}\log j + \frac{1}{2}\log^2j
    + 40X\frac{j}{m} q - 3q
    \le -\log(80) \ .
\end{align}
where we used the definition $X\EqDef \frac{\log
d}{\epsilon}=\frac{\log D_0}{2\epsilon}$ (see \Thm{thm:main}).

To satisfy this equation, we demand that $40X\frac{j}{m} \le 2$ so
that the leading term in the LHS of \Eq{eq:1d-log} would be $-q$.
Substituting $j=q\sqrt{m}$, leads us to define
\begin{align}
  \label{eq:fix-m}
  m\EqDef (20qX)^2 \ .
\end{align}
Going back to \Eq{eq:1d-log}, we now have to satisfy
\begin{align}
  \frac{3}{2}\log j + \frac{1}{2}\log^2j -q \le -\log(80) \ .
\end{align}
where now $j=q\sqrt{m}=20Xq^2$. 

It is easy to see that this equation can always be satisfied for
large enough $q$, since the logarithmic factors are weaker than the
$(-q)$ factor. A straightforward analysis yields
\begin{align}
  q \le \orderof{\log^2 X} \ ,
\end{align}
and therefore
\begin{align}
  \log(\hDel^{-1}) = 3q-1 \le \orderof{\log^2 X} \ ,
\end{align}
and
\begin{align}
  m &= (20Xq)^2 \le \orderof{X^2\log^4 X} \ .
\end{align}

Then $D_I = D_0^{km}$, and as $\log D_0^k =
40qX=\orderof{X\log^2X}$, we get
\begin{align}
  \log D_I = m\log D_0^k = \orderof{X^3\log^6 X} \ .
\end{align}

This concludes the proof of \Lem{lem:diluting}.

\section{2D and beyond}
\label{sec:2D}

Can \Thm{thm:main} be extended to the 2D case and beyond? Currently,
we do not have such a proof. Nevertheless, it is possible to make
a small step in the right direction, as we sketch in this section.

For the sake of clarity, we will restrict ourselves to the 2D case,
and consider the case where the bi-partitioning of the system is
along one dimension, and the length of the boundary is $I$, as
depicted in \Fig{fig:2D-settings}. In terms of the discussion in the
Introduction, $I=|\partial L|$, the area of the separating surface
between $L$ and $\overline{L}$.
\begin{figure}
  \begin{center}
    \includegraphics[scale=1]{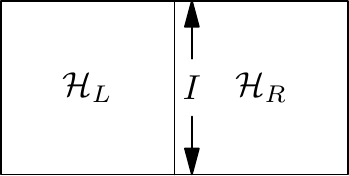}
  \end{center}  
  \caption{\label{fig:2D-settings} A simple settings for a 2D case.
  $I$ is the length of the boundary.}
\end{figure}

To prove an area-law for this system, one would like to show that
the von-Neumann entropy between the two parts of the system satisfies
$S\le \orderof{1}\cdot I$. 

A straightforward approach for obtaining a bound on $S$, which was
also mentioned in \Ref{ref:Has07}, is to treat the 2D system as a 1D
system by considering the particles along a column as a huge
particle of dimension $d^I$. Then to get a bound from \Thm{thm:main}
we replace $d\mapsto d^{I}$, or equivalently, $X\EqDef\frac{\log
d}{\epsilon}\mapsto I\cdot X$. This gives us $S \le \orderof{1}
\cdot (I\cdot X)^3\log^8(I\cdot X)$.

The above derivation completely failed to take into account the
local aspects of the problem along the cut. We now show how one can
make use of it to drop the leading power of $I$ from $I^3$ to $I^2$,
and get a bound of $S\le \orderof{1}\cdot
I^2\cdot X^3\log^8(I\cdot X)$. 

We first note that up to some unimportant geometrical factors, the DL
works also in 2D (and in any dimension for that matter). See
\Ref{ref:DL2} for more details. For simplicity, let us assume that
also in the present case we have only two layers and it is only one
layer that increases the SR with respect to the cut, which we will
still refer to as the ``even layer''. The idea is then to mimic the
1D case and replace a segment of $m$ columns around the cut, which
we denote by $I_m$ with the operator $\hPi_m$. The difference is
that now $I_m$ contains $m\cdot I$ projections instead of $m$.
Therefore, the polynomial that we would use would be
$[C_{mI}(x)]^q$, where $C_{mI}(x)$ is based on the Chebyshev
polynomial of degree $\sqrt{mI}$. Just as in the 1D case, the
shrinking factor $\hDel$ of $\hA$ is given by $\hDel = \Delta_0 +
(1/9)^q$.

What is the SR exponent $\hD$ of this construction? A very similar
analysis to that in \Sec{sec:SR} can be done here: we recursively
divide $I_m$ by cutting it in \emph{parallel} to the cut. Just as in
the 1D case, there are $m$ such possible cuts in $I_m$.  The
difference is that now the SR contribution of the restriction of
$\BN$ to a certain cut is not $D_0$ but $\orderof{1}\cdot I\cdot
D_0$, because $\BN$ contains a \emph{sum} of $\orderof{1}\cdot I$
projections along a cut that is parallel to the boundary. Therefore,
the overall SR after $\ell$ layers is:
\begin{align}
  \SR{\hA^\ell \ket{\phi}} \le D_I \cdot \hD^\ell \ ,
\end{align}
with
\begin{align}
  \hD &\EqDef 20j^{3/2}2^{\frac{1}{2}\log^2 j} (I D_0)^{j/m} \ , \\
  D_I &\EqDef (D_0^{mI}) \ ,
\end{align}
and $j\EqDef q\sqrt{m I}$. In comparison with the 1D indices in
Eqs.~(\ref{eq:pre-DI}, \ref{eq:pre-D}), we see that $D_0 \mapsto
D_0^{I}$ in the formula for $D_I$, but $D_0\mapsto I D_0$ in the
formula for $\hD$ -- an exponential saving in the latter. 

From here, the analysis follows essentially the same steps as the 1D
case. We perform an initial $k$-coarse-graining to
drive $(D_0, \Delta_0) \mapsto (D_0^k, \Delta^k_0)$. Demanding that
$\Delta_0^k \le (1/8)^q$, we set $k\EqDef 20q/\epsilon$ as in
\Eq{eq:fix-k}. Then the analogue of \Eq{eq:1d-log} is
\begin{align}
\label{eq:2d-log}
    \frac{3}{2}\log j + \frac{1}{2}\log^2j
    + (\log I+40Xq)\frac{j}{m} - 3q
    \le -\log(80) \ ,
\end{align}
and demanding that $(\log I+20Xq)\frac{j}{m} \le 2q$, leads us to
set $m\EqDef (\frac{1}{2}\log I + 20Xq)^2 \cdot I$. Plugging this
back to \Eq{eq:2d-log} yields $q\le\orderof{1}\cdot
\log^2(I\cdot X)$, and so $m\le \orderof{1}\cdot I\cdot
X^2\log^4(I\cdot X)$. This results in $\log D_I \le \orderof{1}\cdot
I^2\cdot X^3\log^6(I\cdot X)$, and $\log \hD \le \orderof{q} \le
\orderof{1}\cdot\log^2(I\cdot X)$, which brings us to
\begin{align}
  S \le \orderof{1}\cdot I^2 \cdot X^3\cdot \log^8(I\cdot X) \ .
\end{align}

\section{Matrix Product States}
\label{sec:MPS}

We briefly sketch the implications of these results for the
efficient approximation of the ground state via Matrix Product
States.  Specifically, \Thm{thm:main} implies the following
corollary
\begin{corol}
  Under the same conditions of \Thm{thm:main}, for any integer $k>0$
  there exists a matrix product state (MPS) $\ket{\psi_k}$ with bond
  dimension $k$ such that 
  \begin{align}
  \label{eq:MPS-approximation}
    \norm{ \gs - \psi_k}^2 
      \le 2^{\orderof{1}\cdot X^3\log^8X}(n/k) \ .
  \end{align}
\end{corol}
This result follows from first noting that across \emph{any} cut, we
can use the properties of the AGSP $K$ to bound the norm of the tail
of the Schmidt coefficients: $\sum_{i\ge D^\ell} \lambda_i^2 \le
\frac{1}{\mu^2}\Delta^\ell \le 2D\Delta^\ell$, where $\mu$ is the
overlap of the initial product state with the ground state (see
\Lem{lem:mu1} in \Sec{sec:proof} and the proof of \Lem{lem:S} in the
appendix ). Then letting $k\EqDef D^\ell$, we find that if we
truncate the Schmidt coefficient of a given cut at $k$, we introduce
an error of $\delta=\sum_{i\ge k}\lambda_i^2 \le \frac{2D}{k}$,
where we used the fact that $D\cdot\Delta \le 1/2$.  Then applying
the MPS construction procedure of Vidal~\cite{ref:Vid04a}, and
truncating the SR across all cuts to $k$, yields an MPS of bond
dimension $k$ that approximates the ground state to within the
accumulated error $n\delta = \frac{2nD}{k}$. Finally, recalling from
\Eq{eq:D-bound} in \Lem{lem:diluting} that $\log D\le
\orderof{1}\cdot X^3\log^8X$ gives \Eq{eq:MPS-approximation}.

\section{Conclusions}
\label{sec:Conclusions}

In conclusions we have given a new proof of the area-law for 1D,
gapped and frustration-free systems. The proof uses the DL and the
Chebyshev polynomials to upper bound the entropy by
$\orderof{1}\cdot X^3\log ^8 X$, for $X=\frac{\log d}{\epsilon}$,
which is exponentially better than the bound in \Ref{ref:Has07}. It
brings us much closer to the recent lower bound of
$\orderof{1}\cdot\epsilon^{1/4}$ (for fixed $d$) of Hastings and
Gottesman \cite{ref:lowS-2} and Irani \cite{ref:lowS-1}.

There are two immediate directions in which one might hope to
improve this result. First, it is seems very plausible that the
proof can be extended to the frustrated case. Indeed, already in
Hastings' 1D proof \cite{ref:Has07}, one of the first steps of the
proof is to reduce the frustrated system into an almost
frustration-free system by a coarse-graining procedure. It is
possible that a similar technique can be also deployed here.
Moreover, one might try to take a more direct approach and construct
the AGSP directly from the Hamiltonian $H$, by replacing $\BN$ with
$H$. Both are sums of projections, and a similar SR analysis can
be performed on operators of the form $\poly(H)$. 

The second direction, which is much more interesting, and probably
more difficult, is to try to generalize the area-law for 2D and
beyond. In fact, any sub-volume law for 2D would be an extremely
interesting result. One possibility is to improve the $\log d$
dependence of the bound in \Eq{eq:S}. A bound linear in $\log d$
would imply an area-law in all dimensions, whereas anything like
$(\log d)^{2-\delta}$ would imply a sub-volume law for low
dimensions. However, it seems to us that the right approach is to
better understand and exploit the locality of the problem in the
direction \emph{parallel} to the cut. This was done in a very weak
way in \Sec{sec:2D}, and led to an improved bound in the 2D case. We
believe there are more local aspects of the problem that can be
used. For example, in the current AGSP construction we do not assume
anything about the underlying distribution of violations. Yet, these
arise from a very specific local operation, namely the application
of the previous AGSPs. If one could show that the distribution of
violations decays exponentially in $k$, it may be enough
to use a Chebyshev polynomial of a degree smaller than $\sqrt{m}$ --
which may lead to an area law. More generally, one might want to prove
some notion of independence, or decay of correlation, along the cut,
thereby reducing the 2D problem to a stack of nearly independent 1D
problems.

\section{Acknowledgments}
\label{sec:Acknowledgements}

We are grateful to Dorit Aharonov, Matt Hastings, Sandy Irani,
Tobias Osborne and Bruno Nachtergaele for inspiring discussions
about the above and related topics.  Itai Arad acknowledges support
by Julia Kempe’s ERC Starting Grant QUCO and Julia Kempe’s
Individual Research Grant of the Israel Science Foundation (grant
No. 759/07). Zeph Landau and Umesh Vazirani were supported in part
by ARO grant W911NF-09-1-0440 and NSF Grant CCF-0905626.

%
%

\appendix*


\section{Bounding the entropy of a steps-like probability distribution}
\label{sec:steps}

\begin{proof}[\ of \Lem{lem:S}]

Let $\ket{\gs}=\sum_{i\ge 1} \lambda_i \ket{L_i}\otimes\ket{R_i}$ be
the Schmidt decomposition of the ground state $\ket{\gs}$, and let
$\ket{\phi} = \ket{L}\otimes\ket{R}$ a product state such that 
$|\braket{\phi}{\gs}|=\mu$. Define the sequence of states
$\ket{v_\ell}$ to be the \emph{normalization} of the vectors
$K^{\ell} \ket{\phi}$. Since $K$ is a $(D,\Delta)$-AGSP, 
it follows that $\ket{v_\ell}$
has the following properties:
\begin{enumerate}
  \item  $\SR{v_\ell} \le D^{\ell}$. 
  \item $|\braket{{v_\ell}}{\gs}| 
    \ge \frac{ \mu}{\sqrt{\mu^2 + \Delta^\ell(1-\mu^2)}}$.
\end{enumerate}
Then by Fact~\ref{f:rankapprox} (Eckart-Young theorem),
\begin{align*}
  \sum_{i>D^\ell} \lambda_i^2
   \le  \frac{\mu^2}{\mu^2 + \Delta^\ell(1- \mu^2)}
   \le 1-\frac{\mu^2}{\mu^2+\Delta^\ell}
   \le \frac{1}{\mu^2}\Delta^\ell \EqDef p_\ell\ .
\end{align*}
We will use this bound to upper bound the entropy of the
$\{\lambda_i^2\}$.  Choose $\ell_0= \lceil \frac{\log \mu^2}{\log
\Delta} \rceil$ so that $p_{\ell_0}< 1$.  For $\ell \ge 2 \ell_0
+1$, we we will upper bound the entropy of the distribution with the
bounds
\begin{align*}
  \sum_{j=D^{2\ell}+1}^{D^{2(\ell+1)}} \lambda _j^2 
    \le p_{2\ell}\le  p_\ell \ .
\end{align*}
This is maximized by spreading out the probability mass uniformly in
each interval $[D^{2\ell} + 1, D^{2(\ell+1)}]$ which results in an
upper bound on the contribution of the entropy on the interval $[
D^{2\ell}+ 1, D^{2(\ell+1)}]$ of
\begin{align}
\nonumber
  &p_\ell \log \frac{D^{2(\ell+1)}-D^{2\ell} }{p_\ell} 
    \le p_\ell \log \frac{D^{2(\ell+1)}}{p_\ell}\\
\nonumber
    &\le \Delta^{\ell-\ell_0} \log \frac{
    D^{2(\ell+1)}}{\Delta^{\ell-\ell_0}} \\
    &=\Delta^{\ell-\ell_0}(\ell -\ell_0)
      \log\frac{ D^{2(\ell+1)/(\ell-\ell_0)}}{\Delta} \ .
\label{e:1}
\end{align}
Above, the second inequality follows from $p_{\ell _0} <1$ along
with the fact that the function $x \log \frac{D^{2\ell +1}}{x}$ is
increasing for $x \leq 1$.
 
Summing \Eq{e:1} over $\ell\geq2 \ell_0 +1$ then yields an upper
bound of the entropy contribution of this tail as
\begin{align*}
  &\sum_{\ell \geq 2 \ell_0+1}\Delta^{\ell-\ell_0} 
    (\ell-\ell_0)\log \frac{ D^{2(\ell+1)/(\ell-\ell_0)}}{\Delta}\\
  &\le \sum_{\ell' \geq \ell_0 +1} 
      \ell' \Delta^{\ell'} \log (\frac{D^4}{\Delta} )
    \le \frac{\Delta}{(1-\Delta)^2} \log (\frac{D^4}{\Delta}) \ .
\end{align*}
Above, the first inequality follows from first noting that for the
choice of $\ell$, the exponent on the power of $D$ is upper bounded
by $4$, then making the substitution of $\ell' = \ell - \ell_0$. 
The second follows from using the series equality $\sum _{j\ge
1}jr^j=\frac{r}{(1-r)^2}$.

Combining this estimate with the maximal contribution of entropy
over the first $D^{2(\ell_0 +1)}$ terms, namely $(2\ell_0 +1) \log
D=\orderof{1}\cdot\frac{ \log (\mu ^{-1})}{\log \frac{1}{\Delta}}
\log D$, gives the following bound on the entropy:
\begin{align*}
  S&\le \orderof{1}\cdot
  \frac{\log(\mu^{-1})}{\log(\Delta^{-1})}\log D
    + \frac{\Delta}{(1-\Delta)^2}\log (\frac{D^4}{\Delta}) \ .
\end{align*}  

To simplify this expression, note that for any integer $k\ge 1$,
$K^k$ is an AGSP with characteristic factors $(D^k, \Delta^k)$.
Choosing $k= \lceil \frac{1}{\log(\Delta^{-1})} \rceil$ ensures that
$\frac{1}{4}\le\Delta^k \leq \frac{1}{2}$.  Then substituting the
parameters $D^k$ for $D$ and $\Delta ^k$ for $\Delta$ in the above
expression yields 
\begin{align*}
  S &\le \orderof{1}\cdot\big[\log(\mu^{-1})
    \log D^k + \log{D^{4k}} + 1\big]\\
  &\le\orderof{1}\cdot k \log(\mu^{-1})\log D\\
  &=\orderof{1}\frac{\log(\mu^{-1})}{\log(\Delta^{-1})}
   \log D \ .
\end{align*}

\end{proof}

\bibliography{QC}

\end{document}